\documentclass[11pt,draftcls,onecolumn]{IEEEtran}
\usepackage{amssymb, amsmath, graphicx, cite, color, epsfig}

\newcommand{\reals}{\mathbb{R}}

\newcommand{\iid}   {\mbox{i.i.d.} }
\newcommand{\eqdef}{ := }

\newcommand{\myproof}[1]{\noindent\hspace{2em}{\itshape #1 }}

\newcommand{\EE}{ {\sf E} }

\newcommand{\bfx}{ {\bf x} }

\newcommand{\bfs}{ {\bf s} }
\newcommand{\bfy}{ {\bf y} }
\newcommand{\bfX}{ {\bf X} }
\newcommand{\bfU}{ {\bf U} }

\newcommand{\bfS}{ {\bf S} }
\newcommand{\bfY}{ {\bf Y} }
\newcommand{\bfZ}{ {\bf Z} }

\newcommand{\calX}{ {\cal X} }
\newcommand{\calE}{ {\cal E} }
\newcommand{\calF}{ {\cal F} }
\newcommand{\calM}{ {\cal M} }
\newcommand{\calA}{ {\cal A} }
\newcommand{\calQ}{ {\cal Q} }
\newcommand{\calL}{ {\cal L} }

\newcommand{\calJ}{ {\cal J} }

\DeclareRobustCommand{\1}[1]{\ensuremath \mathbf{1}_{\{#1\}}}

\date{}

\newcommand{\wlim}{\operatorname*{w-lim}}
\newcommand{\pliminf}{\operatorname*{p-liminf}}

\newtheorem{theorem}{Theorem}
\newtheorem{lemma}[theorem]{Lemma}

\newtheorem{definition}[theorem]{Definition}
\newtheorem{corollary}[theorem]{Corollary}

\newtheorem{remark}[theorem]{\it Remark}

\allowdisplaybreaks[0]

\begin{document}
\title{Typical Sequences for Polish Alphabets}
\author{Patrick Mitran,~\IEEEmembership{Member,~IEEE},
\thanks{ The author is with the 
Department of Electrical and Computer Engineering, University of Waterloo, 
Waterloo, Ontario N2L 3G1 (email:~pmitran@ecemail.uwaterloo.ca).}
}

\maketitle

\begin{abstract}
The notion of typical sequences plays a key role in the theory of information.
Central to the idea of typicality is that a sequence $x_1, x_2, \ldots, x_n$
that is $P_X$-typical should, loosely speaking,
have an empirical distribution that is in some sense close to the distribution $P_X$.
The two most common notions of typicality are that of strong (letter) typicality and
weak (entropy) typicality. While weak typicality allows one to apply many arguments
that can be made with strongly typical arguments, some arguments for
strong typicality cannot be generalized to weak typicality.

In this paper, we consider an alternate definition of typicality, namely one based
on the weak* topology and that is applicable to Polish alphabets 
(which includes $\reals^n$). This notion is a generalization of
strong typicality in the sense that it degenerates to strong typicality
in the finite alphabet case, and can also be applied to mixed and continuous distributions.
Furthermore, it is strong enough to prove a Markov lemma, and thus
can be used to directly prove a more general class of results than weak typicality.
As an example of this technique, we directly prove achievability for Gel'fand-Pinsker
channels with input constraints for a large class of alphabets and channels without 
first proving a finite alphabet result and then resorting to 
delicate quantization arguments.

While this large class does not include Gaussian distributions with
power constraints, it is shown to be straightforward to recover this
case by considering a sequence of truncated Gaussian distributions.
\end{abstract}


\section{Introduction}

Perhaps the most intuitive method of deriving achievable rates in information theory
for stationnary memoryless problems is with the concept of
typical sequences. Roughly speaking, a sequence is typical
if its empirical distribution is close, in some sense, to
some ideal distribution.

Ignoring minor variations in definitions, there are essentially
two broad notions of typical sequences. The first notion,
called weakly typical sequences, measures the closeness of a sequence
$\bfx = (x_1, \ldots, x_n)$ to a distribution $P_X$ by quantifying
the probability of the sequence $\bfx$. Specifically, the length $n$ sequence
$\bfx$ is $(P_X, \epsilon)$-weakly typical if
\begin{align}
 \left| \frac{1}{n}\sum_\ell \log P_X(x_\ell) - H(X) \right| < \epsilon,
\end{align}
where $P_X$ is a probability mass function (pmf) and $H(X)$ the entropy of $X$
if $X$ is discrete while $P_X$ is a probability density function (pdf) 
and $H(X)$ the differential entropy of $X$ if $X$ is continuous. 
For this reason, weakly typical sequences are often referred to as entropy typical.
The notion of weak typicality does not appear to generalize well
to mixed distributions.

By contrast, strong typicality characterizes a sequence $\bfx$
by the relative frequency of the occurrence of each letter of the alphabet of
$X$. Specifically, if $N(a|\bfx)$ denotes the number of occurrences
of the letter $a$ in the length $n$ sequence $\bfx$, then $\bfx$ is $(P_X, \epsilon)$-strongly typical
if
\begin{align}
 \left| N(a|\bfx)/n - P_X(a) \right| < \epsilon,
\end{align}
for all $a$\footnotemark. Strong typicality is sometimes referred to by the
more descriptive name of letter typicality. 
Evidently, strong typicality implies at most a countable alphabet.

\footnotetext{In some variations, the additional condition that $N(a|\bfx) = 0$ if $P_X(a) = 0$ is imposed.}

Strong typicality has at least two key consistency properties not shared with weak typicality.
First, strong typicality is sufficient for proving a Markov lemma, which
is a key technique in many network information theory proofs. The Markov
lemma is essentially a corollary of the following broad statement, which one 
would intuitively expect to be true for any reasonable definition of
typicality: {\em if a typical sequence, generated in some arbitrary fashion,
is input to a stationnary memoryless
channel, then the input and output sequences should be jointly typical in
some sense.} Unfortunately, this statement is not possible in general with weak typicality.
We call this desirable property the channel consistency property.

A second desirable property of a typical sequence involves cost functions.
Specifically, let $g:\calX \rightarrow \reals$ be a mapping from the alphabet of
$X$ to the reals. If the length $n$ sequence $\bfx$ is $P_X$-typical, one would expect that the weighted
sum $\frac{1}{n} \sum_\ell g(x_\ell)$ would be reasonably approximated by
$E_{P_X}[g(X)]$. While such a statement can be formalized for strong typicality,
entropy typicality by itself is not sufficient to imply this property\footnotemark. We call this property
the cost consistency property.

\footnotetext{It should be noted that for continuous/discrete distributions,
a variation of so-called distortion typical sequences 
can resolve this issue for a specific cost function $g(x)$.}

Finally, typical sequences should satisfy certain large deviations
results. For example, if $P_{XY}$ is a joint distribution with marginals $P_X$ and $P_Y$, 
$\bfy$ is $P_Y$-typical and $\bfX$ is generated independently of $\bfy$ with each letter
\iid according to $P_X$,
then the probability of $\bfX$ and $\bfy$ being jointly $P_{XY}$-typical should be
on the order of $\approx 2^{-nI(X;Y)}$. This result is usually shown for
strong typical sequences.

The contribution of this paper is to introduce a notion of typicality
based on any metric that induces the weak* convergence of probability measures. We call this
notion weak* typicality. The notion is
sufficiently general to apply to any distribution (discrete, continuous or mixed) 
where the alphabet is a Polish space, and reduces to strong typicality
for finite alphabets. This includes, for example, mixed distributions in
$\reals^n$.

We show that this notion of typicality 
allows for both of the above consistency properties in addition to the
usual rules that one expects a typical sequence to follow, e.g., if
a pair of sequences $(\bfx, \bfy)$ are jointly typical, one expects
that each of $\bfx$ and $\bfy$ are typical. As an example of this
weak* typicality, we directly prove an achievability result
for Gel'fand-Pinsker channels with alphabets in Polish spaces
and cost constraint at the transmitter without having to resort
to delicate quantization arguments which are typically handwaved. Indeed, a key contribution
of this work is that by employing the notion of weak* typicality, one
does not need to directly invoke quantization arguments.

Two important remarks are in order. While, the notion of
weak* typicality avoids the technical difficulties of first
proving a result in the discrete case and then employing
delicate quantization arguments, some of the large deviation results for weak*
typical sequences are proved by using quantization arguments.
However, these arguments are not necessary for the application
of weak* typical sequences. 

Second, for the cost consistency property to apply, 
the cost function must be bounded (and continuous).
This initially precludes weak* typical sequences from
being directly applicable to Gaussian input distributions
with power constraints. However, the result in the Gaussian case
can be recovered by considering a sequence of
truncated Gaussians.

The techniques proved here do not result in more general 
expressions for channel capacity. The most general expression
for channel capacity is given by information spectrum methods
\cite{Han:1994,Han:2003} which apply not only to channels with memory but even
non-ergodic channels. The information spectrum approach though,
looks at quantities such as
\begin{align}
 \underbar{I}(X;Y) \eqdef \pliminf_{n \rightarrow \infty} \frac{1}{n} \log \frac{W^n(\bfY^n | \bfX^n)}{P_{\bfY^n}(\bfY^n)}.
\end{align}
This characterization is based on ratios of probabilities and, similar
to weak typicality, this characterization does not appear to allow
for a Markov lemma. Nevertheless, the results
presented here allow one to generalize 
many results derived using strong typical sequences to a larger class of alphabets
in a straightforward manner.

Weak* convergence of probability measures has found some uses in
the information theory literature. Perhaps the most relevant application to
this work was by Csisz\'ar to compute the capacity of arbitrary
varying channels with general alphabets and states \cite{Csiszar:1992}.
In that work, channels must satisfy a weak* continuity property
and we adopt the same requirement for channels here. In \cite{Csiszar:1992},
the capacity is first computed for the special case of finite
input alphabets and this result is then used to derive the general case.
It should be noted that while weak* convergence of measures plays a key
role in \cite{Csiszar:1992}, no new notion of typical sequences is introduced there.

A key difficulty with continuous alphabets is the analytical
characterization of mutual information. In \cite{Schwarte:1996a}
it was shown that channel capacity is a lower semi-continuous
function in the weak* topology. In \cite{Schwarte:1996b}, sufficient
conditions are found which ensure that channel capacity can be approached
by discrete input distributions or uniform input distribution with finite support
for general alphabet channels.
In \cite{Fozundal:2005},
necessary and sufficient conditions  for weak* continuity and
strict concavity of mutual information were found in addition
to conditions that characterize the capacity value and capacity
achieving measure for channels with side-information at the receiver.
In \cite{Kieffer:1981} and \cite{Kieffer:1982}, Keiffer proves
coding theorems for stationnary (but not necessarily
memoryless) channels that are weak* continuous. While a notion
of typicality appears in \cite{Kieffer:1981} and \cite{Kieffer:1982},
it is a variant of strong typicality and defined only on discrete alphabets.

Weak* convergence has also found applications in analyzing the
stability of recursive algorithms such as variants of the LMS
adaptive filtering algorithm \cite{Bucklew:1993}.

The rest of this paper is structured as follows. In Section
\ref{sec:prelim}, we define channel inputs, input constraints,
channels and information measures such as divergence, as well as some
measure theoretic considerations. In Section \ref{sec:typical}, we define weak*
typicality and prove our key results. In Section \ref{sec:example},
we provide two examples of weak* typical sequences by proving achievability
results for point-to-point and Gel'fand-Pinsker channels. 
In Section \ref{sec:Gaussian}, we discuss the Gaussian case.
In Section \ref{sec:conclusion}, we conclude this work.
The Appendix contains two technical proofs 
for weak* typical sequences.


\section{Preliminaries}
\label{sec:prelim}

\subsection{Alphabets and weak* convergence}
\label{sec:prelim:alpha}

In this paper all measureable spaces are Polish (i.e., complete separable metric spaces) and endowed
with the Borel $\sigma$-field generated by the open sets.
We denote a measureable space by $E$ and the $\sigma$-field
by $\calE$. $\calM_1(E)$ denotes the set of probability measures on $E$.
When we need to distinguish two spaces, we employ
subscripts such as $(E_X, \calE_X)$ and $(E_Y, \calE_Y)$.
Unless clear from context, a random variable $X$ will take values in the alphabet $E_X$
and have distribution $\calL(X)$ which is usually denoted by $P_X$.

$P(A)$ and $E[X]$ denote the probability of event $A$ and mean of random variable $X$
where the underlying measure is always clear from context, or explicitly stated
by subscript.

If $E_X$ and $E_Y$ are Polish, then so is $E_X \times E_Y$.
The corresponding $\sigma$-field $\calE_{XY} \eqdef \calE_X \otimes \calE_Y$ is the
smallest $\sigma$-field containing all rectangles $A \times B$, $A \in \calE_X$, $B \in \calE_Y$.

A sequence $(x_1, \ldots, x_n) \in E_X^n$ will be denoted by $\bfx^n$.
When the length is clear from context or not relevant, we simply write $\bfx$.
An \iid random sequence $\bfX = (X_1, \ldots, X_n)$ consists of a sequence of independent random
variables $X_1, X_2, \ldots, X_n$, each taking values in $E_X$ and for which the 
laws $\calL(X_i) = P_X$ are identical.

In this paper, the notion of weak* convergence of probability
measures plays a key role. We denote the weak* convergence of a sequence of measures
$P_n$ to a limiting measure $P$ by $P = \wlim_{n \rightarrow \infty} P_n$.
The Portemanteau theorem \cite[Theorem 13.16]{Klenke:2008} provides the following equivalent conditions
which will be used in the sequel.
\begin{theorem}
 \label{thm:portemanteau}
 Let $E$ be a Polish space and $P, P_1, P_2, \ldots \in \calM_1(E)$. Then the
following are equivalent
\begin{enumerate}
 \item $P = \wlim_{n \rightarrow \infty} P_n$.
   \label{thm:pmt:wlim}
 \item $\lim_{n \rightarrow \infty} \int f\;dP_n = \int f \;dP$ for all bounded continuous $f$.
   \label{thm:pmt:bc}
 \item $\lim_{n \rightarrow \infty} P_n(A) = P(A)$ for all $A \in \calE$ with $P(\partial A) = 0$,
where $\partial A$ denotes the boundary of $A$.
   \label{thm:pmt:sets}
\end{enumerate}
We note that condition \ref{thm:pmt:bc} is usually taken to be the definition of weak* convergence.
\end{theorem}

\subsection{Channels and Channel Inputs}
\label{sec:prelim:channel}

A memoryless channel from an input $X$ to an output $Y$ is described by
a transition kernel $W_{Y|X}(B|x)$ for $x \in E_X$, $B \in \calE_Y$
which must satisfy the usual measurability conditions of a kernel.
Furthermore, as in \cite{Csiszar:1992}, we make the following
additional continuity assumption on $W_{Y|X}(\cdot|x)$:

\begin{definition}
\label{def:channel}
 A transition kernel $W_{Y|X}(\cdot|x)$ is said to be a channel
if $W_{Y|X}(\cdot|x)$ depends continuously (in the weak* sense)
on $x$, i.e.,
\begin{align}
 W_{Y|X}(\cdot|x) = \wlim_{n \rightarrow \infty} W_{Y|X}(\cdot|x_n),
\end{align}
whenever  $x = \lim_{n} x_n$.
\end{definition}


Given a measure $P_X$ and a kernel $W_{Y|X}$ we denote by
$P_X \otimes W_{Y|X}$ and $P_X W_{Y|X}$ the measures
on $E_X \times E_Y$ and $E_Y$ uniquely defined by
\begin{align}
 (P_X \otimes W_{Y|X})(A \times B) & \eqdef \int_A W_{Y|X}(B|x) P_X(dx) \\
 (P_X W_{Y|X})(B) &\eqdef \int_{E_X} W_{Y|X}(B|x) P_X(dx),
\end{align}
for Borel sets $A \in \calE_X$ and $B \in \calE_Y$.
When clear from context, we will denote the marginal 
$P_X W_{Y|X}$ of $P_X \otimes W_{Y|X}$ by $P_Y$.

In practice, sequences are input to communication channels. In this
paper, all sequences belong to product spaces, i.e.,
an input sequence $\bfx = (x_1, \ldots, x_n)$ belongs
to $\times_{i=1}^n E_X \eqdef E_X^n$ while an output sequence,
say $\bfY^n$, belongs to $\times_{i=1}^n E_Y \eqdef E_Y^n$.

In general, an input sequence $\bfx$ results in a random output sequence $\bfY$
described by a transition
kernel $W_{\bfY|\bfX}(B^n|\bfx)$ where $B^n \in \calE_Y^n$.
In this paper,
all channels are stationnary and memoryless. Thus if a
sequence $\bfx = (x_1, \ldots, x_n)$ is input into to a kernel
$W_{\bfY|\bfX}$, then the probability that the output
$\bfY$ lies in a product set $\times_{\ell=1}^n B_\ell$, $B_\ell \in \calE_Y$, is
$\prod_{\ell=1}^n W_{Y|X}(B_\ell|x_\ell)$, where $W_{Y|X}(\cdot|x)$ is assumed
to satisfy the constraint of Definition \ref{def:channel}.

It is common for channel inputs to be required to satisfy some
constraints. For example, the Gaussian channel typically has
a power constraint
\begin{align}
 \frac{1}{n} \sum_{\ell=1}^n |x_\ell|^2 < P \label{eq:pconstraint}.
\end{align}

We now establish the equivalent concept in this paper.
\begin{definition}
We say that a function
$g(x)$ and a threshold $\Gamma$ form an input constraint 
provided $g(\cdot)$ is continuous and bounded. We say that the
input vector $\bfx = (x_1, \cdots, x_n)$ satisfies the input constraint provided
\begin{align}
 \frac{1}{n} \sum_{\ell=1}^n g(x_\ell) < \Gamma,
\end{align}
and with abuse of notation, we define 
$g(\bfx) \eqdef \frac{1}{n} \sum_{\ell=1}^n g(x_\ell)$. 
\end{definition}

\begin{remark}
A bounded constraint or cost $g(x)$ is sometimes assumed in
the literature \cite{Kramer:2008}. If the input alphabet $E_X$
is compact, then any continuous $g(x)$ is always bounded.
While the bounded assumption on $g(x)$ initially precludes
a power constraint of the form \eqref{eq:pconstraint} when the alphabet is $\reals$, we will
see that the classic result in the additive Gaussian noise case with a power constraint
can be recovered by considering a sequence (in $L$) of
compact input alphabet $E_{X_L}$.
Finally, if the sequence of empirical input distributions on $X$
converges weakly* to $P_X$, then the continuity assumption
can be relaxed to $P_X(D_g) = 0$ where $D_g$ is the
set of discontinuities of $g(\cdot)$ (see part (iii) of \cite[Theorem 13.16]{Klenke:2008}).
\end{remark}

\subsection{Two Examples}

We now give two common examples of alphabets that satisfy the
above constraints.

First, consider two random variables $X$ and $Y$ whose alphabets $E_X$ and $E_Y$
are finite. Then these trivially satisfy the assumptions of 
Section \ref{sec:prelim:alpha} if we choose as metric the
trivial metric $d_X(x_1,x_2) = 1$ if $x_1 \neq x_2$ and 0 otherwise,
and likewise for $d_Y(\cdot,\cdot)$.
Furthermore, in this case, any channel from $X$ to $Y$ satisfies
the weak* continuity assumption of Section \ref{sec:prelim:channel} since
if a sequence $x_n \rightarrow x$, then in fact $x_n = x$ for
all $n$ greater than some $N$, i.e. $W(\cdot|x_n) = W(\cdot|x)$
for all $n > N$.

As a second example, we consider two random variables $X$ and $Y$
whose alphabets are $E_X = E_Y = \reals^N$. With the usual metric on
the real line, these are Polish spaces. Furthermore, suppose that
$Y = X + Z$, where $Z$ is independent of $X$ and has a density $f(z)$, i.e.,
\begin{align}
 W_{Y|X}(B|x) = \int_B f(y-x) \; dy.
\end{align}
Then as shown in \cite[Lemma 2]{Schwarte:1996a}, the channel $W_{Y|X}$ satisfies
the weak* continuity assumption of Section \ref{sec:prelim:channel}.

\section{Information Measures}

\subsection{Definitions}

In this section, we provide some basic background on information
measures and introduce a key result.

Let $P$ and $M$ be two probability measures defined on $E$ and
let $\calQ = \{Q_1, Q_2, \ldots, Q_{|\calQ|} \}$ be a finite (measurable) 
partition of $E$. Then, we define (see \cite[Section 2.3]{Gray:2009})
\begin{align}
 H_{P||M}(\calQ) \eqdef \sum_{i=1}^{|\calQ|} P(Q_i) \log \frac{P(Q_i)}{M(Q_i)}.
\end{align}

The divergence of $P$ with respect to $M$ is defined as (see \cite[Section 5.2]{Gray:2009})
\begin{align}
 D(P||M) = \sup_{\calQ} H_{P||M}(\calQ),
\end{align}
where the supremum is over all finite measurable partitions $\calQ$ of $E$.

Recall that a field $\calF$ has the properties that i) $E \in \calF$,
ii) if $F \in \calF$ then $F^c \in \calF$ and iii) $\calF$ is closed under
{\em finite} unions. In the next subsection, we will construct a 
field that, in addition to generating the $\sigma$-field, i.e., $\calE = \sigma(\calF)$\footnotemark,
has a desirable property. Fields that generate the $\sigma$-field
are of particular interest due to the following two lemmas.

\footnotetext{$\sigma(\calF)$ is the smallest $\sigma$-field that contains the family $\calF$ of sets.} 

\begin{lemma} \cite[Lemma 5.2.2]{Gray:2009}
 \label{lem:Gray}
 Let $(E, \calE)$ be a measurable space, $\calF$ a field that generates $\calE$
and $P$ and $M$ two measures defined on this space. Then
\begin{align}
 D(P||M) = \sup_{\calQ \subset \calF} H_{P||M}(\calQ).
\end{align}
\end{lemma}

The above lemma states that it is sufficient to restrict the finite partitions
to subsets of a generating field. Suppose $(E_X, \calE_X)$ and $(E_Y, \calE_Y)$
are measure spaces, generated by the fields $\calF_X$ and $\calF_Y$
respectively. Then the product $\sigma$-field $\calE_{XY}$ is generated
by the field of rectangles $\calF_{XY} = \calF_X \times \calF_Y$.
Thus, we have the following result (see \cite[Lemma 5.5.1]{Gray:2009}).
\begin{lemma}
\label{lem:Ixy}
Let $P_{XY}$ be a measure on $E_{XY}$ and $P_X$ and $P_Y$ the respective marginals
on $E_X$ and $E_Y$. Then
\begin{align}
 I(X;Y) = D(P_{XY}||P_X \times P_Y) 
        = \sup_{\calQ_X \subset \calF_X, \calQ_Y \subset \calF_Y} H_{P_{XY} || P_X \times P_Y}(\calQ_X \times \calQ_Y).
\end{align}
\end{lemma}

\subsection{A Special Field}


In this subsection, given a measure $P$, we will construct a generating field $\calF_P$ with 
the key property that the $P$-measure of the boundary of any
set in the field is zero. This last property is desirable as it will ensure
that if a sequence of measures $P_n$ converges weakly* to $P$, then
$P_n(A) \rightarrow P(A)$ for all sets $A \in \calF_P$.

While there is no lack of standard constructions for fields that generate the $\sigma$-field,
given any such field $\calF$, there is no guarantee that $P(\partial A) = 0$ for all $A \in \calF$
and thus, it is necessary to construct such a generating field $\calF_P$ specifically
{\em for each limiting measure} $P$.

\begin{lemma}
\label{lem:field}
 Let $P$ by a probability measure defined on a Polish measure space $(E, \calE)$.
 Then there exists a countable family $\calA \subset \calE$ of open sets that
i) generates the Borel $\sigma$-field $\calE$, i.e., $\sigma(\calA) = \calE$, and
ii) $P(\partial A) = 0$ for all $A \in \calA$.
\end{lemma}

\begin{proof}
Since $E$ is Polish, there is a  countably dense subset of $E$, say $E'$. 
For $\calA$ to generate $\calE$, it is sufficient that for each 
$x \in E'$ there is a countably dense subset $R_x$ of $\reals^+$
with each ball $B(x, r) \in \calA$ for all $r \in R_x$. This is because then
each open set of $E$ is the countable union of balls in $\calA$.

It thus remains to be shown that the sets $R_x$ can be chosen
such that each ball $B(x, r)$ has $P(\partial B(x,r)) = 0$.
For $r > 0$, let $F_x(r) = P(B(x,r))$. Then
$F_x(r)$ is a non-decreasing, bounded below by 0 and above by 1.
Thus it has left and right limits and at most a countable number of jump discontinuities
and thus, there is a countably dense subset $R'_x$ of $\reals^+$
for which $F_x(r)$ is continuous. We claim that choosing $R_x = R'_x$
will do. Since $\partial B(x,r) \subset \{y : d(x, y) = r\}$,
then $P(\partial B(x,r)) \leq F_x(r+) - F_x(r)$ where
$F_x(r+)$ is the right limit of $F_x(r)$. However, for $r \in R'_x$, 
$F_x(r+) = F_x(r)$ by continuity.
\end{proof}

\begin{corollary}
 \label{cor:field} 
  Let $P$ be a probability measure defined on a Polish measure space $(E, \calE)$.
 Then there exists a countable generating field $\calF_P$ with the property that $P(\partial A) = 0$
for all $A \in \calF_P$.
\end{corollary}
\begin{proof}
 Since  $\partial(A \cap B) \subset \partial A \cup \partial B$,
 $\partial(A \cup B) \subset \partial A \cup \partial B$ and 
 $\partial A = \partial(A^c)$, one can
 extend the countable family $\calA$ in Lemma \ref{lem:field} to include all finite 
 intersections/unions of balls in $\calA$ and complements of balls in $\calA$, i.e.,
 extend $\calA$ to a field $\calF_P$.
\end{proof}

\section{Weak* Typical Sequences}
\label{sec:typical}

\subsection{Definitions}


Given a sequence $\bfx = (x_1, \ldots, x_n) \in E_X^n$, one can associate
an empirical distribution $P_{\bfx}$ defined by
\begin{align}
 P_{\bfx}(A) \eqdef \frac{1}{n} \sum_{\ell=1}^n \1{x_\ell \in A}.
\end{align}
When clear from context, we denote the empirical distribution
by $P_X^n$. Likewise, given two sequences $\bfx = (x_1, \ldots, x_n) \in E_X^n$
and $\bfy = (y_1, \ldots, y_n) \in E_Y^n$, the joint empirical distribution
$P_{\bfx,\bfy}$ is defined by
\begin{align}
 P_{\bfx,\bfy}(A \times B) \eqdef \frac{1}{n} \sum_{\ell=1}^n \1{x_\ell \in A} \1{y_\ell \in B},
\end{align}
which is denoted by $P_{XY}^n$ when clear from context.

A sequence $\bfx$ should be typical if
its empirical distribution is in some sense close to some
probability measure. In this paper, closeness is measured
with respect to the weak* topology.

Specifically, let $d(\cdot, \cdot)$ be any metric on the
space of probability measures $\calM_1(E)$ that induces
the weak* topology and fix this metric for the rest of the
paper. The Prohorov metric is an example of
such a metric and the exact choice of the metric is irrelevant.

We denote by $B(M, \epsilon)$ the ball of distributions
$\{P \in \calM_1(E) : d(P,M) < \epsilon \}$.
We will say that an empirical
distribution $P^n$ is $P$-typical when its distance from
$P$ is sufficiently small. We make the following definitions.

\begin{definition}
 Let $d(\cdot,\cdot)$ be a metric on the
space of probability measures $\calM_1(E_X)$ that induces
the weak* topology. A sequence $\bfx = (x_1, \ldots, x_n)$
is said to be weakly* $(P_X,\epsilon)$-typical if $d(P_{\bfx}, P_X) < \epsilon$.
Similarly, we say that an empirical distribution $P_X^n$
is weakly* $(P_X,\epsilon)$-typical if $d(P_X^n, P_X) < \epsilon$.
We denote the set of such length $n$ typical sequences by 
$A^n_\epsilon(P_X)$.
\end{definition}

\begin{remark}
 For finite alphabets, it is interesting to compare the definition of weak* typicality
 to that of strong typicality. Specifically, if $|E_X|$ is finite and
 $P_\bfx \in A^n_\epsilon(P_X)$, then $|P_\bfx(a) - P_X(a)| < \delta$ for some $\delta > 0$ and 
 all $a \in E_X$, and $\delta \rightarrow 0$ as $\epsilon \rightarrow 0$. This
 coincides with the definition of strong typicality (except
 for the occasional requirement that $P_\bfx(a) = 0$ if $P_X(a) = 0$).
 Thus weak* typical sequences can be viewed as a generalization of strong typical sequences.
\end{remark}

Unless stated otherwise, in the sequel all typical sequences are weak* typical sequences.

We also find it convenient to introduce a notion of {\em asymptotically typical
sequences}. Given a set of sequences $\bfx^{n_1}, \bfx^{n_2}, \cdots$
of lengths $n_1 < n_2 < \cdots$, there
is a corresponding sequence of empirical distributions
$P_X^{n_1}, P_X^{n_2}, \cdots$. 
\begin{definition}
 We say that the sequence of sequences $\{\bfx^{n_k}\}$ is
asymptotically $P_X$-typical provided that the corresponding
sequence of empirical measures satisfies
\begin{align}
 P_X = \wlim_{k \rightarrow \infty} P_X^{n_k}.
\end{align}
\end{definition}

\begin{remark}
If a sequence of sequences $\{\bfx^{n_k}\}$ is
asymptotically $P_X$-typical, then for
any $\epsilon > 0$ there exists a $K$ such that for
all $k>K$, $\bfx^{n_k}$ is weak* $(P_X,\epsilon)$-typical. 
\end{remark}

It some cases it will be more convenient to first prove certain results
for asymptotically typical sequences of sequences, and then as a corollary infer
behavior of typical sequences for large length $n$.

Jointly weak* typical sequences are defined analogously. Specifically,
\begin{definition}
  Two sequences $\bfx = (x_1, \ldots, x_n)$ and $\bfy = (y_1, \ldots, y_n)$
are said to be weak* $(P_{XY},\epsilon)$-typical if $d(P_{\bfx, \bfy}, P_{XY}) < \epsilon$.
Similarly, we say that the empirical distribution $P_{XY}^n$
is weak* $(P_{XY},\epsilon)$-typical if $d(P_{XY}^n, P_{XY}) < \epsilon$.
We denote the set of such pairs of length $n$ weak* typical sequences by $A^n_\epsilon(P_{XY})$.
\end{definition}

Likewise, one can also define a sequence of a pair of sequences $\{(\bfx^{n_k}, \bfy^{n_k})\}$
to be asymptotically $P_{XY}$-typical in the obvious way.

\subsection{Consistency Properties}

There are several desirable properties that typical and jointly typical sequences
should possess.

First, a random $\iid$ sequence should be typical with high probability.
Second, a $(P_X,\epsilon)$-typical sequence $\bfx$ should have a cost
$\frac{1}{n} \sum_i g(x_i)$ close to $E_{P_X} [g(X)]$. Third,
if two sequences are jointly typical, then one would expect each sequence to be typical in its own right.

The following lemma shows that the first is indeed true for asymptotically typical sequences.

\begin{lemma} \label{lem:w*b}
%
 Let $X_1, X_2, \ldots$ be a sequence of independent random variables with values 
 in $E_X$ with identical distribution $P_X$ and $\{\bfX^{n_k}\}$ a corresponding sequence of sequences.
 Then almost surely $\{\bfX^{n_k}\}$ is asymptotically $P_X$-typical. 
 \label{lem:w*b:Varadarajan}
%
%
%
\end{lemma}
\begin{proof}
This is a direct restatement of Varadarajan's Theorem \cite[Theorem 11.4.1]{Dudley:2003}.
\end{proof}

We now show that all three statements are true for weak* typical sequences,
where the first is a consequence of the result for asymptotically typical sequences.

\begin{theorem} \label{thm:w*b}
 The following hold.
\begin{enumerate}
  
 \item Let $X_1, X_2, \cdots, X_n$ be independent random variables with values in $E_X$ with identical distribution
 $P_X$. Then for any $\epsilon > 0$,
 \begin{align}
  \lim_{n \rightarrow \infty} P\left( \bfX^n \in A_\epsilon^n(P_X) \right) = 1.
 \end{align}

 \label{thm:w*b:Varadarajan}
 
  \item For every $\delta > 0$, there is an $\bar{\epsilon}(\delta) > 0$ such 
       that if $\bfx \in A^n_{\bar{\epsilon}(\delta)}(P_X)$ then
   \begin{align}
   \left| E_{P_\bfx} [g(X)] - E_{P_X}[g(X)] \right| < \delta.
   \end{align}
 \label{thm:w*b:cost}
 
 \item For any $\epsilon > 0$, there is an $\bar{\epsilon}(\epsilon) > 0$ such 
       that if $(\bfx, \bfy) \in A^n_{\bar{\epsilon}(\epsilon)}(P_{XY})$ then
       $\bfx \in A^n_{\epsilon}(P_X)$.

 \label{thm:w*b:marginal}

\end{enumerate}
\end{theorem}
\begin{proof}
 i) Otherwise one could find a sequence of sequences $\{\bfX^{n_k}\}$ in 
Lemma \ref{lem:w*b} that would not be asymptotically typical almost surely. 

ii) Let $M_X$ be any (not necessarily empirical)
measure on $E_X$. By part \ref{thm:pmt:bc} of Theorem \ref{thm:portemanteau}, there is
an $\bar{\epsilon}(\delta)$ such that $M_X \in B(P_X, \bar{\epsilon}(\delta))$ implies
   \begin{align}
   \left| E_{M_X} [g(X)] - E_{P_X}[g(X)] \right| < \delta.
   \end{align}
The result follows since $\bfx \in A^n_{\bar{\epsilon}(\delta)}(P_X)$ iff $P_{\bfx} \in B(P_X, \bar{\epsilon}(\delta))$.

iii) By part \ref{thm:pmt:sets} of Theorem \ref{thm:portemanteau}, if $M^k_{XY}$ is a sequence
of (not necessarily empirical) 
measures in $E_{XY}$ with marginals $M^k_X$ and such that $\lim_k d(M^k_{XY}, P_{XY}) = 0$,
then $\lim_k d(M_X^k, P_X) = 0$. Thus for each $\epsilon$, there is an $\bar{\epsilon}(\epsilon)$
such that $d(M_{XY}, P_{XY}) < \bar{\epsilon}(\epsilon)$ implies $d(M_X, P_X) < \epsilon$. 
The result again follows since $(\bfx, \bfy) \in A^n_{\bar{\epsilon}(\epsilon)}(P_{XY})$ iff 
$P_{\bfx,\bfy} \in B(P_{XY}, \bar{\epsilon}(\epsilon))$.
\end{proof}

An important desirable property of typical sequences is that if a 
typical sequence is input to a channel then the input and
output should be jointly typical in some sense. We have the following
theorem.
\begin{lemma}
 \label{lem:aml}
 Let the sequence of sequences $\{\bfx^{n_k}\}$ be asymptotically
 $P_X$-typical and consider the output sequences $\bfY^{n_k}$ generated
 by the stationary memoryless channel $W_{Y|X}(\cdot|x)$. Then
 the sequence of sequences $\{(\bfx^{n_k}, \bfY^{n_k})\}$ is
 asymptotically $P_X \otimes W_{Y|X}$-typical almost surely.
\end{lemma}

\begin{remark}
Consider the Markov chain $X - Y - Z$.
Then if a sequence of sequences $\{(\bfx^n, \bfy^n)\}$
is asymptotically $P_{XY}$-typical and
is used to generate output sequences $\bfZ^n$ according
to the channel $W_{Z|XY}(\cdot|x,y) = W_{Z|Y}(\cdot|y)$,
then the sequences $\{(\bfx^n, \bfy^n, \bfZ^n)\}$
are almost surely asymptotically $P_{XY} \otimes W_{Z|Y}$-typical,
i.e., the Markov Lemma holds for asymptotically typical sequences.
\end{remark}

\begin{proof}
 For ease of notation, we consider the case $n_k = k$.
 Let $P^k_X$ denote the empirical distribution of the $k$-length
sequence $\bfx^k$. Let $P_Y$ denote the marginal $P_X W_{Y|X}$
and let $P_{XY}^k$ denote the empirical distribution of the pair
of sequences $\bfx^k$ and $\bfY^k$. 

The outline of the proof is as follows. First, for each of
$P_X$ and $P_Y$, consider two fields $\calF_X$ and $\calF_Y$
as described in Corollary \ref{cor:field}.
We will first show that almost surely for all $A \in \calF_X$, $B \in \calF_Y$,
$\lim_k P_{XY}^k(A \times B) = P_X \otimes W_{Y|X}(A \times B)$.
Thus, by \cite[Chapter 1, Theorem 2.2]{Billingsley:1999},
$P_X \otimes W_{Y|X} = \wlim_k P_{XY}^k$ since each open set of
$E_X \times E_Y$ is a countable union of rectangular sets $A \times B$,
$A \in \calF_X$, $B \in \calF_Y$.


Now, consider a set $A \in \calF_X$ and $B \in \calF_Y$ and observe that
\begin{align}
 & \left| P_{XY}(A \times B) - P_{XY}^k(A\times B) \right| \nonumber \\
 & \quad \leq \left| P_{XY}(A\times B) - P^k_X \otimes W_{Y|X} (A \times B) \right|
    + \left| P^k_X \otimes W_{Y|X} (A \times B) - P^k_{XY} (A \times B) \right|,
 \label{eq:cons1}
\end{align}
where $P_{XY} = P_X \otimes W_{Y|X}$.
From Lemma 2 of \cite{Csiszar:1992}, since $P_X = \wlim_k P_X^k$,
then $P_X \otimes W_{Y|X} = \wlim_k P_X^k \otimes W_{Y|X}$.
Since 
$P_X \otimes W_{Y|X}(\partial(A \times B)) \leq P_X(\partial A) + P_Y(\partial B) = 0$,
it follows that $\lim_k P_X^k \otimes W_{Y|X}(A \times B) = P_X \otimes W_{Y|X}(A \times B)$
and thus the first term on the right side of \eqref{eq:cons1} is 0 in the limit.

As for the second term on the right side of \eqref{eq:cons1}, we note that
\begin{align}
 P^k_X \otimes W_{Y|X} (A \times B) - P^k_{XY} (A \times B)
 &= \frac{1}{k} \sum_{i=1}^k \1{x_i \in A} \left[ W_{Y|X}(B|x_i) - \1{Y_i \in B} \right]
\end{align}
Let $Z_i = \1{x_i \in A} \left[ W_{Y|X}(B|x_i) - \1{Y_i \in B} \right]$. Then,
the second term on the right of \eqref{eq:cons1} is
\begin{align}
\left| \frac{1}{k} \sum_{i=1}^k Z_i \label{eq:cons_sum} \right|.
\end{align}
Note that given the non-random sequence $\bfx^k$, i) the random variables $Z_i$ are independent, 
ii) $E[Z_i] = 0$, and iii) $\sup_i var[Z_i] \leq 1$ since $-1 \leq Z_i \leq 1$.
Then by \cite[Theorem 5.29]{Klenke:2008}, the sum in \eqref{eq:cons_sum}
converges to 0 almost surely, i.e.,
\begin{align}
 \lim_{k \rightarrow \infty} \left| P^k_X \otimes W_{Y|X} (A \times B) - P^k_{XY} (A \times B) \right| = 0
\end{align}
almost surely. Since the family of sets $\calF_X$ and $\calF_Y$ are countable,
it follows that almost surely the right side of \eqref{eq:cons1} vanishes 
as $k \rightarrow \infty$ for all $A \in \calF_X$ and $B \in \calF_Y$.

\end{proof}

\begin{theorem}
\label{thm:aml}
Let $\bfx^n$ be an input sequence to a stationnary memoryless channel $W_{Y|X}(\cdot|x)$
and let $\bfY^n$ be the corresponding output sequence.
For every $\epsilon > 0$ and $\delta > 0$, there exists an $\bar{\epsilon}(\epsilon,\delta)$ 
such that if $\bfx^n \in A^n_{\bar{\epsilon}(\epsilon, \delta)}(P_X)$ for all $n$ greater than some $N$, then
\begin{align}
 \liminf_{n \rightarrow \infty} P\left( (\bfx^n, \bfY^n) \in A^n_\epsilon(P_X \otimes W_{Y|X}) \right) > 1 - \delta. 
\end{align}
\end{theorem}
\begin{proof}
 Suppose that for a given $\epsilon$ and $\delta$, no such $\bar{\epsilon}(\epsilon, \delta)$ can
be found. Then, we can find a sequence of sequences $\{\bfx^{n_k}\}$ with corresponding empirical measures
$P^{n_k}_X$ for increasing $n_k$ such that $d(P^{n_k}_X, P_X)<1/k$, i.e., $P_X = \wlim_k P^{n_k}_X$ and
\begin{align}
 \liminf_{k \rightarrow \infty} P\left[(\bfx^{n_k}, \bfY^{n_k}) \in A^{n_k}_\epsilon(P_X \otimes W_{Y|X}) \right] \leq 1-\delta.
\end{align}
But this contradicts the almost sure asymptotic $(P_X \otimes W_{Y|X})$-typicality 
of $(\bfx^{n_k}, \bfY^{n_k})$ in Lemma \ref{lem:aml}.
\end{proof}

\subsection{Large Deviations}

The next theorems provide some large deviations results for weak* typical
sequences. The first theorem looks at the probability of a random \iid
sequence drawn according to a law $P$ to be weak* $M$-typical. This
is shown to be $\approx 2^{-nD(M||P)}$ which is the same result as for
weak and strong typical sequences. For example, let $P_{XY}$ be the joint law
for $E_X \times E_Y$ with marginals $P_X$ and $P_Y$. Then the probability that
a pair of sequences $\bfX$ and $\bfY$ drawn according to $P_X \otimes P_Y$
is $P_{XY}$-typical is $\approx 2^{-nI(X;Y)}$.

The next two theorems then consider the more special case of when the sequence
$\bfy$ is non-random and known to be weak* typical but $\bfX$ is random.
There, we again show that the probability that the pair of sequences is
weak* typical is $\approx 2^{-nI(X;Y)}$. This result is normally proved
for strong typical sequences using the notion of conditional strongly typical sequences
and no analog exists in general for weak typical sequences.

\begin{theorem}
 \label{thm:dev_rand}
 Let $P$ and $M$ be measures on a common probability space $(E, \calE)$ and let 
a random sequence $\bfX^n$ be chosen \iid according to the law $\calL(X_i)=P$, i.e., 
draw the sequence $\bfX^n$ according to the measure $\mu^n = \otimes_{i=1}^n P$. 
Define the sequence of probabilities $a_n = \mu^n(\bfX^n \in A^n_\epsilon(M))$,
i.e., the probability that the drawn sequence is weak* $(M,\epsilon)$-typical.
If $D(M||P)$ is finite, then there is an $\epsilon(\delta)>0$ such that for all $\epsilon<\epsilon(\delta)$
\begin{align}
 -D(M || P) \leq \liminf \frac{1}{n} \log a_n \leq \limsup \frac{1}{n} \log a_n \leq -D(M || P) + \delta.
 \label{eq:dev_rand}
\end{align}
If $D(M || P) = \infty$, then for each $L > 0$ there is an $\bar{\epsilon}(L)$ such that
for $\epsilon < \bar{\epsilon}(L)$, the right side of \eqref{eq:dev_rand} is $-L$.
\end{theorem}

\begin{remark}
 \label{rem:delta}
 Under these assumptions, it follows that for any $\bar{\delta} > 0$, and $\delta > 0$,
there is a sufficiently large $N$ such that for all $n > N$,
\begin{align}
 \mu^n(\bfX^n \in A^n_{\epsilon}(M)) \leq 2^{n(-D(M||P) +\delta + \bar{\delta})}.
\end{align}
Since both $\delta$ and $\bar{\delta}$ are arbitrary, they can be absorbed
into a single ``$\delta$'' term.
\end{remark}

\begin{proof}
 Let $P_X^n$ be the empirical measure of the drawn sequence.
 We first show the lower bound. By the definition of weak* typicality, we recognize that
\begin{align}
a_n = \mu^n(\bfX^n \in A^n_\epsilon(M)) = \mu^n( P^n_X \in B(M, \epsilon)),
\end{align}
where $B(M, \epsilon)$ is an open ball in the space $\calM_1(E)$ in the weak* topology.
Then, by Sanov's Theorem (Corollary 6.2.3 of \cite{Dembo:1998}) we have the large deviations
principle 
\begin{align}
  -\inf_{\nu \in B(M, \epsilon)} \Lambda^*(\nu) \leq \liminf \frac{1}{n} \log \mu^n( P^n_X \in B(M, \epsilon)).
\end{align}
The lower bound then follows since in the weak* topology
$\Lambda^*(\nu) = D(\nu || P)$ (Lemma 6.2.13 of \cite{Dembo:1998}) and 
the $\inf$ is bounded by selecting any choice of $\nu \in B(M, \epsilon)$, say $\nu = M$.

To prove the upper bound, we use a quantization argument.
Consider a field $\calF_M$ as described in Corollary \ref{cor:field}
for $M$. 

If $D(M||P)$ is finite, then $M \ll P$, and
for any $\delta>0$, choose a sufficiently fine finite partition $\calQ \subset \calF_M$ of $E$ such that
the induced discrete probabilities $Q_P$ and $Q_M$ of $P$ and $M$ on the atoms of $\calQ$
satisfy\footnotemark
\begin{align}
 D(Q_M || Q_P) \geq D(M || P) - \delta/2.
\end{align}
\footnotetext{Since $D(M||P)$ is finite, this is straightforward by Lemma \ref{lem:Gray}.} 
Otherwise, $D(M||P) = \infty$ and for each $L > 0$, we can find a partition $\calQ$ such that
$D(Q_M || Q_P) > 2L$.


In either case, denote these atoms by $A_1, \ldots, A_K$ for some integer $K$.
Let $Q^n_P$ be the induced discrete probability of the empirical measure $P_X^n$ on the 
atoms of $\calQ$. Then, the event $P_X^n \in B(M, \epsilon)$ implies that the
discrete probabilities $|Q_M(A_k) - Q^n_P(A_k)| < \delta_1$
for all atoms and $\delta_1 \rightarrow 0$ as $\epsilon \rightarrow 0$
by weak* convergence since $M(\partial A_k) = 0$.

Now, $Q^n_P$ is an empirical probability for a random sequence over a finite alphabet. Let
$\Gamma$ be the set of all probability distributions $Q_P$ on the atoms of $\calQ$
such that $|Q_M(A_k) - Q_P(A_k)| \leq \delta_1$ for all $k$. In the finite alphabet case 
we have the following well-known large deviations result (\cite[Theorem 2.1.10]{Dembo:1998})
\begin{align}
 \limsup \frac{1}{n} \log \mu^n( P_X^n \in B(M, \epsilon)) \leq 
 \limsup \frac{1}{n} \log \mu^n( Q^n_P \in \Gamma) \leq - \inf_{Q_\nu \in \Gamma} D(Q_\nu || Q_P).
\end{align}
If $D(M||P)$ is finite, then $M \ll P$, and the divergence $D(Q_\nu || Q_P)$ is continuous on the 
compact set of $Q_\nu$
such that $Q_\nu \ll Q_P$ which includes $Q_M$ and $D(Q_\nu || Q_P)$ 
is infinite otherwise. Thus for small enough $\delta_1 > 0$, the 
$\inf$ can be bounded by $D(Q_M || Q_P) - \delta_2$ where $\delta_2 \rightarrow 0$ as $\delta_1 \rightarrow 0$. Hence, pick $\epsilon$ small enough that $\delta_2 < \delta/2$.

If $D(M||P) = \infty$, there are two cases to consider.

First, if $Q_M \ll Q_P$, then the same argument as the finite $D(M||P)$ case above shows that
\begin{align}
 -\inf_{Q_\nu \in \Gamma} D(Q_\nu || Q_P) \leq -2L + \delta_2,
\end{align}
where $\delta_2 < L$ can be ensured by choosing $\epsilon$ small enough. 

Second, if we do not have $Q_M \ll Q_P$, then there is a set $A \in \calQ$
such that $Q_M(A) > 0$ and $Q_P(A) = 0$. If $\epsilon$ is chosen sufficiently small that
$\delta_1 < M(A)/2$, then $Q_\nu(A) > 0$ for all $Q_\nu \in \Gamma$ and $D(Q_\nu || Q_P) = \infty$ 
for all $Q_\nu \in \Gamma$.

Either way, $-\inf_{Q_\nu \in \Gamma} D(Q_\nu || Q_P) \leq -L$, where
$L>0$ is arbitrary.
\end{proof}

\begin{theorem}
\label{thm:nr_ld_ub}
Let $P_{XY}$ be a joint distribution on $E_X \times E_Y$ and $P_X$ and
$P_Y$ denote its marginals.
Let $\bfy^{n}$ be a sequence 
and $\bfX^n$ a random sequence drawn \iid according to $\mu^k = \otimes_{i=1}^n P_X$. 
If $D(P_{XY} || P_X \times P_Y)$ is finite then for each $\delta > 0$, there
are $\epsilon(\delta)$ and $\bar{\epsilon}(\delta)$ such that if $\epsilon < \epsilon(\delta)$,
$\bar{\epsilon} < \bar{\epsilon}(\delta)$ and
$\bfy^n \in A_{\bar{\epsilon}}^n(P_Y)$ for all $n$ greater than some $N$, then
\begin{align}
 \limsup_n \frac{1}{n} \log \mu^{n}( (\bfX^{n}, \bfy^{n}) \in A^n_{\epsilon}(P_{XY}) ) 
 \leq -D(P_{XY} || P_X \times P_Y) + \delta.
 \label{eq:nr_ld_ub}
\end{align}
If $D(P_{XY} || P_X \times P_Y) = \infty$, then for every $L > 0$, there is
a sufficiently small $\epsilon, \bar{\epsilon} > 0$ such that \eqref{eq:nr_ld_ub} holds with the
right side replaced by $-L$.
\end{theorem}
\begin{proof}
 See Appendix.
\end{proof}

\begin{theorem}
\label{thm:nr_ld_lb}
Let $P_{XY}$ be a joint distribution on $E_X \times E_Y$ and $P_X$ and
$P_Y$ denote its marginals.
Let $\bfy^{n}$ be a sequence and $\bfX^{n}$ a random sequence drawn \iid 
according to $\mu^{n} = \otimes_{i=1}^{n} P_X$. 
Then, for each  $\delta > 0$ and $\epsilon>0$  there is an $\bar{\epsilon}(\epsilon, \delta)>0$
such that if $\bfy^n \in A^n_{\bar{\epsilon}(\epsilon, \delta)}(P_Y)$ for all $n$ greater than some $N$, then
\begin{align}
 \liminf_n \frac{1}{n} \log \mu^{n}( (\bfX^{n}, \bfy^{n}) \in A^n_{\epsilon}(P_{XY}) ) \geq - D(P_{XY} || P_X \times P_Y) - \delta.
\end{align}
\end{theorem}
\begin{proof}
See Appendix.
\end{proof}

\section{Examples}
\label{sec:example}

We now apply the notion of weak* typical sequences to prove achievability
results for two channel coding examples. The first is the traditional
point-to-point channel. While more general results can be obtained
using information spectrum methods, the example highlights the
application of weak* typical sequences.

In the second example, we apply weak* typical sequences to
Gel'fand-Pinsker channels. These results cannot be obtained 
for arbitrary Polish spaces using weak/strong typical sequences.

In this section, the cost constraint $g(x)$ is continuous and
bounded. In Section \ref{sec:Gaussian}, we will consider the
Gaussian case with power constraint.

\subsection{Point-to-Point Channel}

We consider communicating over a channel $W_{Y|X}$, where the alphabets
$E_X$ and $E_Y$ are Polish spaces. For completeness, we briefly state some definitions.

An $(n, M, P_e)$ code is a set of $M$ codewords $\bfx_1, \ldots, \bfx_M$
and a decoder $\phi: E_Y^n \rightarrow \{1, \ldots, M\}$ such that the 
average probability of error is  
\begin{align}
 P_e = \frac{1}{M} \sum_{v=1}^M P[\phi(\bfY) \neq v | \bfX = \bfx_v].
\end{align}

A rate $R$ is said to be achievable if there is a
sequence of codes $(n, M_n, P_e^n)$ with block lengths 
$n$, $R = \lim_n \frac{1}{n} \log M_n$, 
and probability of error $P_e^n \rightarrow 0$.

We will show the following well known result using weak* typical sequences.

\begin{theorem}
Let $W_{Y|X}$ be a communication channel with Polish input and output alphabets and 
input constraint $(g(x), \Gamma)$. Then any rate 
\begin{align}
 R < \sup_{P_X: E_{P_X} [g(X)] < \Gamma} I(X; Y)
\end{align}
is achievable.
\end{theorem}
\begin{remark}
 The converse can be obtained with the usual Fano inequality. 
\end{remark}

\begin{proof}
 The proof follows the usual random coding argument with the exception that we now use
the results derived for weak* typical sequences. Pick any $\gamma > 0$. We will bound the 
probability of error for a random code by $3\gamma$ for large enough $n$.

In particular, as usual, pick a $P_X$ which satisfies the constraint $E_{P_X}[g(X)] < \Gamma$.
We generate $M_n = \lfloor 2^{nR} \rfloor$ codewords of length $n$ with each entry $\iid$ according
to $P_X$, and denote these as $\bfX_1, \ldots, \bfX_{M_n}$.

The encoder transmits $\bfX_V$ where $V$ is uniform among the indices $\{1, \ldots, M_n\}$.
The decoder employs weak* typical decoding. Specifically, it looks for an index
$v$ such that $(\bfX_v, \bfY)$ are weak* $(P_X \otimes W_{Y|X}, \epsilon)$ typical
for some $\epsilon> 0$ and declares $v$ as the transmitted index if such a $v$
exists and is unique. Otherwise, an error is declared.

By the usual symmetry, without loss of generality we assume that the index $v=1$
is selected at the transmitter. The probability of error for a $(P_{X} \otimes W_{Y|X}, \epsilon)$-typical
decoder is then bounded as
\begin{align}
 P_e^n \leq P[g(\bfX_1) \geq \Gamma] + P[(\bfX_1, \bfY) \notin A^{n}_\epsilon (P_X \otimes W_{Y|X})]
+ P[\cup_{v \neq 1} (\bfX_v, \bfY) \in A^{n}_\epsilon (P_X \otimes W_{Y|X})].
\label{eq:cc:bound}
\end{align}

By part \ref{thm:w*b:Varadarajan} of Theorem \ref{thm:w*b},
$P_{e,n}^2 \eqdef P[(\bfX_1, \bfY) \notin A^{n}_\epsilon (P_X \otimes W_{Y|X})] \rightarrow 0$ as $n \rightarrow 0$.
Thus $P_{e,n}^2 < \gamma$ for all $n$ larger than some $N_2$.

Second, we note that by Theorem \ref{thm:dev_rand}\footnotemark and Remark \ref{rem:delta}, 
for any index $v \neq 1$ and any $n$ greater than some sufficiently large $\bar{N}_3$,
\begin{align}
P[(\bfX_v, \bfY) \in A^{n}_\epsilon (P_X \otimes W_{Y|X})] \leq 2^{-n(I(X;Y) - \delta)}
\end{align}
and $\delta \rightarrow 0$ as $\epsilon \rightarrow 0$. Thus, 
for large enough $n$, the union bound implies
\begin{align}
 P_{e,n}^3 \eqdef P[\cup_{v \neq 1} (\bfX_v, \bfY) \in A^{n}_\epsilon (P_X \otimes W_{Y|X})]
 \leq 2^{nR} 2^{-n(I(X;Y)- \delta)},
\end{align}
and $P_{e,n}^3 < \gamma$ for all $n$ larger than some $N_3$ provided $R < I(X;Y) - \delta$.

\footnotetext{Here we assume that $I(X;Y)$ is finite. The case that
$I(X;Y) = \infty$ can be considered separately.}

Finally, one can upper bound $P_{e,n}^1 \eqdef P[g(\bfX_1) \geq \Gamma]$ by
$P[\bfX_1 \notin A^{n}_{\alpha}(P_X)] + P[g(\bfX_1) \geq \Gamma | \bfX_1 \in A^{n}_{\alpha}(P_X)]$
for any arbitrary $\alpha>0$. 

By part \ref{thm:w*b:cost} of Theorem \ref{thm:w*b} and since $E_{P_X} [g(X)] < \Gamma$, there is
a sufficiently small $\alpha>0$ such that 
$P_{e,n}^5 \eqdef P[g(\bfX_1) \geq \Gamma | \bfX_1 \in A^{n}_{\alpha}(P_X)] = 0$.

By part \ref{thm:w*b:Varadarajan} of Theorem \ref{thm:w*b}, for any $\alpha > 0$,
$P_{e,n}^4 \eqdef P[\bfX_1 \notin A^{n}_{\alpha}(P_X)]$ vanishes as $n \rightarrow \infty$.
Thus $P_{e,n}^4 < \gamma$ for all $n$ larger than some $N_4$.

Thus for any rate $R < I(X;Y) - \delta$, the bound in \eqref{eq:cc:bound}
is at most $3\gamma$ for all $n$ sufficiently large.
Finally, $\delta > 0$ can be made arbitrarily small by choosing $\epsilon$ small enough.
\end{proof}

\begin{remark}
 Since $E_{P_X} [g(X)] < \Gamma$ and each codeletter of each codeword is \iid, 
one could have bounded $P[g(\bfX_1) \geq \Gamma]$ by the strong
law of large numbers. However, this approach will not be possible
for Gel'fand-Pinsker channels as the channel input is not generated
by independent and randomly chosen codeletters.
\end{remark}

\subsection{Gel'fand-Pinsker Channels}

We now consider proving an achievability result for Gel'fand-Pinsker channels
assuming Polish alphabets. The achievability result for $R < I(U;Y) - I(U;S)$
was proved in the discrete case in \cite{Gelfand:1980}.  The Gaussian case 
with additive interference and noise was considered in \cite{Costa:1983}
and further results on additive interference and noise can be found in
\cite{Cohen:2002,Yu:2001,Zamir:2002,Kim:2004}. Here, we consider
achievability for a general channel $W_{Y|SX}$ with Polish alphabets directly using weak* typical sequences.

We start with a brief set of definitions. A source sends a message
$V \in \{1, \ldots, M\}$ selected uniformly at random to a receiver by transmitting a sequence
$\bfx$. The channel $W_{Y|XS}$ results in an output $\bfY$ that
depends stochastically on the input $\bfx$ as well as an interference
sequence $\bfS$, where $\bfS$ is an \iid random sequence drawn according
to $P_S$. Furthermore, the encoder is aware of the interference
sequence $\bfS$ apriori and the decoder is unaware of the interference. 
Thus, the encoder is described by the mapping
$\phi^n_{\sf tx}: \{1, \ldots, M\} \times E_S^n \rightarrow E_X^n$
while the decoder is the mapping 
$\phi^n_{\sf rx}: E_Y^n \rightarrow \{1, \ldots, M\}$.

A code is a tuple $(\phi^n_{\sf tx}, M, \phi^n_{\sf rx}, P_e)$
where $P_e = P[\phi^n_{\sf rx}(\bfY) \neq V | \bfX = \phi^n_{\sf tx}(V, \bfY)]$.
A rate $R$ is said to be achievable if there exists
a sequence of codes $(\phi^{n}_{\sf tx}, M_n, \phi^{n}_{\sf rx}, P_e^n)$
with $\lim_n \frac{1}{n} \log M_n = R$ and $\lim_n P_e^n = 0$.

We have the following achievability result for Polish alphabets.
\begin{theorem}
For Gel'fand-Pinsker channels with cost constraint $(g(x), \Gamma)$
at the transmitter, any rate
\begin{align}
 R < \sup_{P_{U|S}, W_{X|US} : E[g(X)] < \Gamma} I(U;Y) - I(U;S),
\end{align}
is achievable where the supremum is over all {\em transition kernels} $P_{U|S}$ and all
{\em channels} $W_{X|US}$.
\end{theorem}
\begin{remark}
 Recall that a channel is a transition kernel that satisfies a weak* continuity condition.
\end{remark}

\begin{proof}
Again, the random coding argument is followed, however we now invoke weak* typical
sequences. Pick any $\gamma > 0$. We will show that for any $n$ larger than some
$N$, the probability of error with a random codebook is at most $4\gamma$.

Specifically, pick a transition kernel $P_{U|S}$ and channel $W_{X|US}$
which satisfy the constraint $E[g(X)] < \Gamma$
and let $\delta > 0$.
First, construct $M_n = \lfloor 2^{nR} \rfloor$ bins, with each bin containing
$\lceil 2^{n(I(U;S)+\delta)} \rceil$ sequences of length $n$ with each 
codeletter generated \iid according to the marginal $P_U$.  We denote
these sequences as $\bfU_1, \bfU_2, \ldots$, $\bfU_K$ where $K = \lfloor 2^{nR} \rfloor \times \lceil 2^{n(I(U;S)+\delta)} \rceil$.

Following the usual argument,
to encode message $v \in \{1, \ldots, M_n\}$, the encoder
looks in bin $v$ for a sequence $\bfU_i$ such that $(\bfU_i, \bfS) \in A_{\epsilon_1}^{n}(P_{US})$,
i.e. $(\bfU_i, \bfS)$ are weakly* $(P_{US},\epsilon_1)$ typical for some appropriate $\epsilon_1 > 0$.

If there is no such $\bfU_i$ sequence, an error is declared, which is denoted by
the event $\EE_1$. Otherwise,
the encoder constructs a sequence $\bfX$ generated by the memoryless channel
$W_{X|US}$. If 
\begin{align}
 \frac{1}{n} \sum_{\ell=1}^n g(X_\ell) \geq \Gamma,
\end{align}
then the transmission of $\bfX$ would violate the channel input constraint
and an error is declared, denoted by the event $\EE_2$. Otherwise,
$\bfX$ is transmitted over the channel.

The receiver obtains $\bfY$ and looks in all bins for a pair of sequences
$(\bfU_j, \bfY)$ that are jointly $(P_{UY}, \epsilon_3)$-typical for some 
appropriate $\epsilon_3$. If there is
a unique such pair, then the bin index in which $\bfU_j$ is present is declared
as the estimate $\hat{v}$ of $v$. If the index is not unique, or the bin
incorrect, we denote this error event by $\EE_3$.

The probability of error is bounded by
\begin{align}
 P_e \leq P[\EE_0] + P[\EE_1 \cap \bar{\EE}_0] + P[\EE_2|\bar{\EE}_1] + P[\EE_3|\bar{\EE}_1],
\label{eq:gp_bound}
\end{align}
where $E_0$ is the event that $\bfS$ is not $(\epsilon_0, P_S)$-typical
for some suitable $\epsilon_0 > 0$.

{\em Error analysis:} 

We start by analyzing $P[\EE_3|\bar{\EE}_1]$ and note that
\begin{align}
P[\EE_3|\bar{\EE_1}] \leq P[\EE_4|\bar{\EE}_1] + P[\EE_5|\bar{\EE}_1, \bar{\EE}_4],
\end{align}
where $\EE_4$ is the event that $(\bfU_i, \bfY)$ are not jointly $(P_{UY}, \epsilon_3)$-typical
and $\EE_5$ is the event that there is an index $j \neq i$ such that $(\bfU_j, \bfY)$  are jointly
$(P_{UY}, \epsilon_3)$-typical.

By applying Theorem \ref{thm:aml} twice, first to the channel $W_{X|US}$ and then to the channel $W_{Y|XS}$
and using Part \ref{thm:w*b:marginal} of Theorem \ref{thm:w*b},
there is an $\bar{\epsilon}_{us}(\epsilon_3, \gamma/2)$ such that for any 
$\bar{\epsilon}_{us} \leq \bar{\epsilon}_{us}(\epsilon_3, \gamma/2)$,
if $(\bfU_i,\bfS)$ is $(P_{US}, \bar{\epsilon}_{us})$-typical, 
then $P_{e,n}^4 \eqdef P[\EE_4|\bar{\EE}_1] < \gamma/2$ for all  $n$ greater than some $N_4$.

Now, conditioned on $\bar{\EE}_4$, by Part \ref{thm:w*b:marginal} of Theorem \ref{thm:w*b},
$\bfY = \bfy$ is $(P_Y, \bar{\epsilon}_Y)$-typical,
where $\bar{\epsilon}_Y \rightarrow 0$ as $\epsilon_3 \rightarrow 0$.
By Theorem \ref{thm:nr_ld_ub}, for any $\delta_3 > 0$, there is
sufficiently small $\epsilon_3(\delta_3)$, $\bar{\epsilon}_3(\delta_3)$ and large $\bar{N}_5$ such that
provided $\epsilon_3 < \epsilon_3(\delta_3)$ and $\bfy \in A^n_{\bar{\epsilon}_3(\delta_3)}$ for $n > \bar{N}_5$,
by the usual union bound
\begin{align}
 P[\EE_5|\bar{\EE}_1, \bar{\EE}_4] \leq 2^{nR} \lceil 2^{n(I(U;S)+\delta)} \rceil 2^{-n(I(U;Y)-\delta_3)},
\end{align}
Thus, selecting $\epsilon_3 < \epsilon_3(\delta_3)$ and $\epsilon_3$ small enough
that $\bar{\epsilon}_Y < \bar{\epsilon}_3(\delta_3)$, 
$P_{e,n}^5 \eqdef P[\EE_5|\bar{\EE}_1, \bar{\EE}_4] < \gamma/2$ for all $n$
greater than some $N_5$ provided $R < I(U;Y) - I(U;S) - \delta_3 - \delta$.
Note that $\delta_3 > 0$ and $\delta > 0$ are arbitrary.

We now analyze $P[\EE_2|\bar{\EE}_1]$.
For any $\epsilon_2 > 0$, there is an $\epsilon_1(\epsilon_2, \gamma) > 0$
such that for $\epsilon_1 < \epsilon_1(\epsilon_2, \gamma)$ by Theorem \ref{thm:aml}, 
\begin{align}
P_{e,n}^2 \eqdef P[(\bfU, \bfS, \bfX) \notin A^n_{\epsilon_2}(P_{USX})] < \gamma
\end{align}
for all $n$ larger than some $N_2$.
By  Part \ref{thm:w*b:cost} of Theorem \ref{thm:w*b},
there is small enough $\epsilon_2$
such that $P[g(\bfX) \geq \Gamma \; | \; (\bfU, \bfS, \bfX) \in A_{\epsilon_2}^n(P_{USX})] = 0$.
Select $\epsilon_1 < \min \{ \bar{\epsilon}_{us}(\epsilon_3, \gamma/2), \epsilon_1(\epsilon_2, \gamma) \}$.

We now analyze $P[\EE_1 \cap \bar{\EE}_0]$. Let $\delta_1 > 0$ be arbitrary. By Theorem \ref{thm:nr_ld_lb},
there is an $\bar{N}_1$, $\epsilon_0 = \epsilon_0(\epsilon_1, \delta_1)$
such that conditioned on the fact that if
$\bfS = \bfs$ is $(P_S, \epsilon_0)$-typical, 
\begin{align}
 P[(\bfU, \bfs) \in A_{\epsilon_1}^n(P_{US})] \geq  2^{-n(I(U;S)+\delta_1)},
\end{align}
for $n > \bar{N}_1$.
Thus,
\begin{align}
 P[\EE_1 \cap \bar{\EE}_0] &= E_{P_\bfS} \left[ P[\EE_1|\bfS = \bfs] \1{\bfS \in A_{\epsilon_0}^n(P_S)} \right]\\
 &= E_{P_\bfS} \left[ \left[ 1 - P[(\bfU, \bfs) \in A_{\epsilon_1}^{n}(P_{US})] \right]^{2^{n(I(U;S) + \delta)}} 
\1{\bfS \in A_{\epsilon_0}^n(P_S)} \right] \\
&\leq E_{P_\bfS} \left[ \left[ 1 - 2^{-n(I(U;S)+\delta_1)} \right]^{2^{n(I(U;S) + \delta)}} \1{\bfS \in A_{\epsilon_0}^n(P_S)} \right] \\
&\leq \left[ 1 - 2^{-n(I(U;S)+\delta_1)} \right]^{2^{n(I(U;S) + \delta)}} \\
&\leq \exp \left( - 2^{-n (\delta_1 - \delta)} \right),
\end{align}
where $P_{\bfS} = \otimes_{i=1}^n P_S$.
Thus, selecting $\delta_1 < \delta$ yields $P_{e,n}^1 \eqdef P[\EE_1|\bar{\EE}_0] < \gamma$ 
for all $n$ greater than some $N_1$.

Finally, we analyze $P[\EE_0]$.
However, by Part \ref{thm:w*b:Varadarajan} of Theorem \ref{thm:w*b},
for any $\epsilon_0 > 0$, $P_{e,0}^n \eqdef P[\bfS \notin A^n_{\epsilon_0}(P_S)]$ vanishes as $n \rightarrow 0$
and thus $P_{e,0}^n < \gamma$ for all $n$ greater than some $N_0$.
\end{proof}

\begin{remark} 
We note the following remarks.
First, we could not rely on the law of large numbers to argue that
$\bfX$ would satisfy the constraint pair $(g(x), \Gamma)$. This is because
while $\bfX$ was generated stochastically, it was done so based on the
pair $(\bfU_i, \bfS)$, where $\bfU_i$ was specifically chosen to satisfy a given property.
Thus instead we argued via Theorem \ref{thm:aml} that the triple 
$(\bfU_i, \bfS, \bfX)$ is $\epsilon_2$-typical (by the channel consistency property or Markov Lemma)
and then that $\bfX$ must satisfy the power constraint for sufficiently small $\epsilon_2$.

Second, we again employed the channel consistency property (or Markov Lemma)
to prove that the pair $(\bfU_i, \bfY)$ are jointly typical with high probability.
\end{remark}

\section{The Gaussian Case}
\label{sec:Gaussian}

In Section \ref{sec:example}, achievability results were proved for
a point-to-point channel as well as the Gel'fand-Pinsker channel with
input constraints. It was noted that due to the input constraint $(g(x), \Gamma)$,
either the (continuous) cost function $g(x)$ should be bounded, or the input alphabet
is compact (which trivially implies $g(x)$ is bounded).

This rules out the consideration of an input constraint $(g(x) = x^2, \sigma^2_X)$ with
a Gaussian input distribution as neither the cost function nor the input
alphabet is then bounded.

In this section, we show how one can recover the traditional achievability
results in both cases for Gaussian distributions. Specifically, we
will consider an input alphabet over the interval $E_{X_L} = [-L, L]$
and show that as $L \rightarrow \infty$, one can arbitrarily approach
the well-known results in the Gaussian case. It should be noted that
we consider all alphabets as subsets of $\reals$ for simplicity of
exposition only and the arguments apply equally well to alphabets
over $\reals^n$.

\subsection{Point-to-Point Channel}
\label{sec:P2P}

Here the capacity of the channel $Y = X + Z$ with $Z \sim {\cal N}(0,\sigma^2_Z)$
is well known to be $C = I(Y; X)$ evaluated for $X \sim {\cal N}(0,\sigma^2_X)$\footnotemark.

\footnotetext{While $X \sim {\cal N}(0,\sigma^2_X)$ does not satisfy the input constraint,
$X \sim {\cal N}(0,\sigma^2_X-\epsilon); 0 < \epsilon < \sigma^2_X$ does and the capacity follows by a simple
limit argument.}

Now, 
consider the family of random variables $X_L$ (indexed by $L > 0$)
with densities
\begin{align}
 f_{X_L}(x) = \left \{
\begin{array}{cc}
 0 & x < -L \\
 f_X(x) / K(L)  & -L \leq x \leq L \\ 
 0 & x > L
\end{array}
\right.,
\label{eq:truncated}
\end{align}
where $f_X(x)$ is the PDF of an ${\cal N}(0,\sigma^2_X)$ distribution and
\begin{align}
 K(L) = \int_{-L}^{L} f_X(x) \; dx
\end{align}
is a normalization constant with the property that $\lim_{L\rightarrow\infty} K(L) = 1$. 

For notational convenience, let the output random variable be denoted by $Y_L$ when the input
is $X_L$, i.e., $Y_L = X_L + Z$, and let the output random
variable be $Y$ when the input is $X \sim {\cal N}(0,\sigma^2_X)$, i.e.,
$Y = X + Z$. It
is straightforward to verify that $E[X_L^2] < P$ for all $L$,
and thus this input distribution satisfies the input constraint,
and the rate
\begin{align}
 R_L = I(Y_L; X_L)
\end{align}
is achievable for each $L$. 
We will show that as $L \rightarrow \infty$,
that $R_L = I(Y_L; X_L) \rightarrow I(X;Y)$.

First, note that 
\begin{align}
I(Y_L ; X_L) &= h(Y_L) - h(Y_L|X_L)  \\
&= h(Y_L) - h(Z),
\end{align}
and thus it suffices to show that $\lim_{L\rightarrow\infty} h(Y_L) = h(Y)$.

Second,
\begin{align}
 f_{Y_L}(y) &= \int_{-\infty}^{\infty} f_{Y|X}(y|x) f_{X_L}(x) \; dx \\
&= \frac{1}{K(L)} \int_{-L}^{L} f_{Y|X}(y|x) f_{X}(x) \; dx
\end{align}
Define 
\begin{align}
 g_{Y_L}(y) &\eqdef \int_{-L}^{L} f_{Y|X}(y|x) f_{X}(x) \; dx,
\label{eq:gYL}
\end{align}
then $f_{Y_L}(y) = g_{Y_L}(y)/K(L)$, and
\begin{align}
 h(Y_L) &= \int_{-\infty}^{\infty} f_{Y_L}(y) \log 1 / f_{Y_L}(y) \; dy \\
 &= \frac{1}{K(L)} \left[ \int_{-\infty}^{\infty} g_{Y_L}(y) \log 1/g_{Y_L}(y) \; dy \right]
 - \int_{-\infty}^{\infty} f_{Y_L}(y) \log K(L) \; dy \\
 &= -\frac{1}{K(L)} \left[ \int_{-\infty}^{\infty} g_{Y_L}(y) \log g_{Y_L}(y) \right] \; dy - \log K(L).
\end{align}
Since $\lim_{L\rightarrow\infty} K(L) = 1$, it remains only to show that the
term in the square brackets converges to $-h(Y)$ as $L \rightarrow \infty$.
However, because of \eqref{eq:gYL}, $g_{Y_L}(y)$
is continuous, strictly positive, strictly increasing in $L$ and converges pointwise to $f_{Y}(y)$.
Thus, the integrand $g_{Y_L}(y) \log g_{Y_L}(y)$ is continuous and 
converges pointwise to  $f_{Y}(y) \log f_{Y}(y)$ as $L \rightarrow \infty$.

Let $A = \{ y \in \reals | f_Y(y) < e^{-1}\}$. Since $x \log x$ is decreasing in $x$
for $0 \leq x < e^{-1}$, then for $y \in A$, $f_{Y_L}(y) \log f_{Y_L}(y)$ is decreasing
in $L$ and 
\begin{align}
 \lim_{L \rightarrow \infty} \int_A g_{Y_L}(y) \log g_{Y_L}(y) \; dy = \int_A f_{Y}(y) \log f_{Y}(y) \; dy
\end{align}
by the (Lebesgue) monotone convergence theorem.

Now consider the set $B = A^c = \{ y \in \reals | f_Y(y) \geq e^{-1}\}$ and note
that $B$ is closed (since $f_Y(y)$ is continuous) and bounded (since $f_Y(y)$
is a Gaussian pdf) and thus $B$ is compact. Thus, on the set $B$, $f_{Y_L}(y)$
converges uniformly to $f_Y(y)$ by Dini's theorem. Hence, there is a large enough
$\bar{L}$ such that for all $L > \bar{L}$, $f_{Y_L}(y) > e^{-2}$ for all $y \in B$. 
Let $K = \sup_{y \in \reals} f_Y(y)$. Then for $y \in B$ and $L > \bar{L}$, 
$|g_{Y_L}(y) \log g_{Y_L}(y)| \leq f_Y(y) \max\{2, |\log K|\}$ which is integrable.
Thus, by the dominated convergence theorem,
\begin{align}
 \lim_{L \rightarrow \infty} \int_B g_{Y_L}(y) \log g_{Y_L}(y) \; dy = \int_B f_{Y}(y) \log f_{Y}(y) \; dy,
\end{align}
and therefore
\begin{align}
 \lim_{L\rightarrow\infty} \int_{-\infty}^{\infty} g_{Y_L}(y) \log g_{Y_L}(y) \; dy
 &= \lim_{L \rightarrow \infty} \int_A g_{Y_L}(y) \log g_{Y_L}(y) \; dy 
+ \lim_{L \rightarrow \infty} \int_B g_{Y_L}(y) \log g_{Y_L}(y) \; dy \\ 
 &= \int_{-\infty}^{\infty} f_{Y}(y) \log f_{Y}(y) \; dy \\
 &= -h(Y),
\end{align}
as desired.

\subsection{Gel'fand-Pinsker Channels}
\label{sec:GP}

In the Gaussian case, it is well-known that the capacity is obtained
with the choice
\begin{align}
 U &= X + \alpha S \label{eq:gp1} \\
 Y &= X + Z + S, \label{eq:gp2}
\end{align}
where $X \sim {\cal N}(0, \sigma^2_X)$ and independent of $S$, and $\alpha$ is an appropriately chosen constant.

We follow the same strategy as in Section \ref{sec:P2P}. Namely, we consider
the family of truncated Gaussians $X_L$ given in \eqref{eq:truncated}, and obtain
\begin{align}
 U_L &= X_L + \alpha S \label{eq:gp_tr1} \\
 Y_L &= X_L + Z + S.\label{eq:gp_tr2}
\end{align}
As previously, we will show that
\begin{align}
 \lim_{L\rightarrow\infty} I(U_L;S) &= I(U;S) \\
 \lim_{L\rightarrow\infty} I(U_L;Y_L) &= I(U;Y). 
\end{align}

First, note that $I(U_L;S) = h(U_L) - h(X_L)$. Second,
$\lim_{L\rightarrow\infty} h(U_L) = h(U)$ follows by the same argument as
in Section \ref{sec:P2P} with $U_L$ replaced by $Y_L$ and
$\alpha S$ replaced by $Z$. Third, $\lim_{L\rightarrow\infty} h(X_L) = h(X)$
since
\begin{align}
\lim_{L\rightarrow\infty} \int_{-L}^L f_X(x) \log f_X(x) \; dx = \int_{-\infty}^{\infty} f_X(x) \log f_X(x) \; dx.
\end{align}

Next, we note that $I(U_L;Y_L) = h(U_L) + h(Y_L) - h(U_L, Y_L)$ and
$\lim_{L\rightarrow\infty} h(U_L) = h(U)$ was already argued, and $\lim_{L\rightarrow\infty} h(Y_L) = h(Y)$
follows similarly. We provide an outline of $\lim_{L\rightarrow\infty} h(U_L, Y_L) = h(U,Y)$.
Because of the additive nature \eqref{eq:gp1} -- \eqref{eq:gp_tr2}, if
$f_{UY|X}(u,y|x)$ denotes the conditional PDF of $U$ and $Y$ given $X$, then
\begin{align}
 f_{U_L,Y_L}(u,y) &= \frac{1}{K(L)} \int_{-L}^{L} f_{UY|X}(u,y|x) f_X(x) \; dx.
\end{align}
Thus, following the same argument as in Section \ref{sec:P2P}, one can
apply the monotone and dominated convergence theorems and 
obtain $\lim_{L\rightarrow\infty} h(U_L, Y_L) = h(U,Y)$.

\section{Conclusion}
\label{sec:conclusion}

In this paper, a notion of typical sequences based on the weak* topology
was defined. This notion of typical sequence applies to discrete,
continuous and mixed distribution and was shown to satisfy
consistency properties normally associated with strongly typical sequences.
As examples of applying these notions of typical sequences, achievable
rates were proved for the traditional point-to-point channel and
Gel'fand-Pinsker channels with Polish alphabets and input constraints.

\appendix

\myproof{Proof of Theorem \ref{thm:nr_ld_ub}:}
 Pick two fields $\calF_X$ and $\calF_Y$ as described in Corollary \ref{cor:field}, 
 and let $\calQ_X \subset \calF_X$ and $\calQ_Y \subset \calF_Y$ be two partitions
 of size $|\calQ_X| = N_X$ and $|\calQ_Y| = N_Y$ and elements $\calQ_X = \{A_1, \ldots, A_{N_X}\}$,
 $\calQ_Y = \{B_1, \ldots, B_{N_Y}\}$.

 Let $P^n_{XY}$ be the empirical measure induced by the pair of sequences $(\bfX^n,\bfy^n)$.
 Let $Q^{n}_{XY}$, $Q^{n}_{X}$ and $Q^n_{Y}$ denote the empirical measures 
 and empirical marginals induced on the partitions $\calQ_X \times \calQ_Y$.
 Furthermore,
 for $B_j$ such that $Q^{n}_Y(B_j) > 0$, define the conditional measure 
 $Q^n_{X|Y}(A_i|B_j) = Q^n_{XY}(A_i \times B_j) / Q^n_{Y}(B_j)$, otherwise
 $Q^n_{X|Y}(A_i|B_j)$ is arbitrary.

 Likewise, starting with $P_{XY}$, let $Q_{XY}$, $Q_X$ and $Q_{Y}$ denote the induced
 measures and marginals on the partitions $\calQ_X \times \calQ_Y$.
 For $B_j$ such that $Q_Y(B_j) > 0$, define $Q_{X|Y}(A_i|B_j) = Q_{XY}(A_i \times B_j) / Q_Y(B_j)$,
 and note that $Q_{X|Y}(.|B_j) \ll Q_X(.)$. Otherwise pick $Q_{X|Y}(A_i|B_j)$
 to be some arbitrary distribution for which $Q_{X|Y}(.|B_j) \ll Q_X(.)$.
 
 Now, pick an $\epsilon_1 > 0$ and note that 
 $(\bfX^{n}, \bfy^{n}) \in A^n_\epsilon(P_{XY})$ implies
 $P^n_{XY} \in B(P_{XY}, \epsilon)$ which for sufficiently small $\epsilon$, 
 itself implies that for all $i$ and $j$,
 \begin{align}
  |Q^n_{XY}(A_i \times B_j) - Q_{XY}(A_i \times B_j)| < \epsilon_1. \label{eq:key}
 \end{align}
Therefore, with this choice of $\epsilon$, 
\begin{align}
\mu^{n}\left( (\bfX^{n}, \bfy^{n}) \in A^n_\epsilon(P_{XY})  \right)
 \leq \mu^n\left( \bigcap_{i,j} \left\{|Q^n_{XY}(A_i \times B_j) - Q_{XY}(A_i \times B_j)| < \epsilon_1 \right\} \right).
 \label{eq:bound}
\end{align}
Let $\calJ_{>0}$ denote the set of $j$ such that $Q_Y(B_j) > 0$
and select $\bar{\epsilon}(\delta) > 0$ small enough such that
$\bfy^n \in A_{\bar{\epsilon}(\delta)}^n(P_Y)$ implies
\begin{align}
|Q^n_Y(B_j) - Q_Y(B_j)| < \epsilon_1 \quad \forall j \label{eq:close1}\\ 
1 - \epsilon_1 < \frac{Q^n_Y(B_j)}{Q_Y(B_j)} \leq 1 + \epsilon_1 \quad \forall j \in \calJ_{>0}. \label{eq:close2}
\end{align}

Now, by \eqref{eq:close1}, for all $j \notin \calJ_{>0}$ and all $i$, \eqref{eq:key} is satisfied
and the right side of the bound in \eqref{eq:bound} can be limited to the intersection of all $i$ and 
all $j \in \calJ_{>0}$. Furthermore, since for $j \in \calJ_{>0}$, \eqref{eq:key} implies
\begin{align} 
 &\left|Q^n_{X|Y}(A_i|B_j)Q^n_Y(B_j) - Q_{X|Y}(A_i|B_j)Q^n_Y(B_j)\right. \nonumber \\
 &\quad + \left. Q_{X|Y}(A_i|B_j)Q^n_Y(B_j) - Q_{X|Y}(A_i|B_j)Q_Y(B_j)\right| < \epsilon_1,
\end{align}
together with \eqref{eq:close1}, this implies
\begin{align}
 \left|Q^n_{X|Y}(A_i|B_j) - Q_{X|Y}(A_i|B_j)\right| Q^n_Y(B_j) < 2\epsilon_1,
\end{align}
and for $j \in \calJ_{>0}$, with \eqref{eq:close2} this implies
\begin{align}
 \left|Q^n_{X|Y}(A_i|B_j) - Q_{X|Y}(A_i|B_j) \right| < \frac{2\epsilon_1}{Q_Y(B_j)(1-\epsilon_1)}.
\end{align}
Let $\epsilon_2 = \max_{j \in \calJ_{>0}} \frac{2\epsilon_1}{Q_Y(B_j)(1-\epsilon_1)}$. Then
$\epsilon_2 \rightarrow 0$ as $\epsilon_1 \rightarrow 0$ and for all $i$ and $j \in \calJ_{>0}$,
\begin{align}
 \left| Q^n_{X|Y}(A_i|B_j) - Q_{X|Y}(A_i|B_j) \right| < \epsilon_2.
\end{align}
Therefore, we have shown that
\begin{align}
\mu^{n}\left( (\bfX^{n}, \bfy^{n}) \in A^n_\epsilon(P_{XY})  \right)
 &\leq \mu^n\left( \cap_{i,j \in \calJ_{>0}} \left\{ |Q^n_{X|Y}(A_i|B_j) - Q_{X|Y}(A_i|B_j)| 
       < \epsilon_2 \right\} \right) \\
 &= \prod_{j \in \calJ_{>0}} \mu^n\left( \cap_{i} \left\{ |Q^n_{X|Y}(A_i|B_j) - Q_{X|Y}(A_i|B_j)| < \epsilon_2 \right\} \right).
 \label{eq:bound2}
\end{align}
Now, let $N_{n,j} = \sum_{\ell=1}^n \1{y_\ell \in B_j}$ for any $j \in \calJ_{>0}$, 
i.e., the number of letters of $\bfy^n$ that are in $B_j$. Then for a given $j \in \calJ_{>0}$,
 \begin{align}
  \limsup_{n \rightarrow \infty} & \frac{1}{n} \log 
    \mu^n\left( \cap_{i} \left\{ |Q^n_{X|Y}(A_i|B_j) - Q_{X|Y}(A_i|B_j)| < \epsilon_2 \right\} \right) \nonumber \\
  &= \limsup_{n \rightarrow \infty} \frac{N_{n,j}}{n} \times \limsup_{n \rightarrow \infty} \frac{1}{N_{n,j}} \log 
    \mu^n\left( \cap_{i} \left\{ |Q^n_{X|Y}(A_i|B_j) - Q_{X|Y}(A_i|B_j)| < \epsilon_2 \right\} \right) \\
  &\leq -(1 - \epsilon_1) Q_Y(B_j) \times \left[ D(Q_{X|Y}(\cdot|B_j) || Q_X(\cdot) ) - \delta_{1,j} \right],
 \end{align}
where $\delta_{1,j} \rightarrow 0$ as $\epsilon_2 \rightarrow 0$ since $Q_{X|Y}(\cdot|B_j) \ll Q_X(\cdot)$
and we have used Theorem 2.1.10 of \cite{Dembo:1998}.
Therefore, 
\begin{align}
 \limsup_{n \rightarrow \infty} \frac{1}{n} \log 
    \mu^{n}\left( (\bfX^{n}, \bfy^{n}) \in A^n_\epsilon(P_{XY})  \right)
 &\leq -(1 - \epsilon_1) D(Q_{XY} || Q_X \times Q_Y) +\delta_2 \\
 &= -(1 - \epsilon_1) H_{P_{XY} || P_X \times P_Y}(\calQ_X \times \calQ_Y) +\delta_2,
\end{align}
where $\delta_2 = (1-\epsilon_1) \sum_{j \in \calJ_{>0}} \delta_{1,j}$.

If $D(P_{XY} || P_X \times P_Y)$ is finite, the result then follows by first 
choosing appropriate fine quantizers $\calQ_X$ and $\calQ_Y$ such that
\begin{align}
 -H_{P_{XY} || P_X \times P_Y}(\calQ_X \times \calQ_Y) < -D(P_{XY} || P_X \times P_Y) + \delta/2,
\end{align}
and then choosing $\epsilon_1$ small enough (thus $\epsilon(\delta)$ and $\bar{\epsilon}(\delta)$ small enough)
such that
\begin{align}
-(1 - \epsilon_1) H_{P_{XY} || P_X \times P_Y}(\calQ_X \times \calQ_Y) +\delta_2 \leq 
-D(P_{XY} || P_X \times P_Y) + \delta.
\end{align}

If $D(P_{XY} || P_X \times P_Y) = \infty$ then for every $L > 0$, we can find a pair of quantizers
such that $H_{P_{XY} || P_X \times P_Y}(\calQ_X \times \calQ_Y) > 2L$. The result follows
by choosing $\epsilon_1$ small enough (thus $\epsilon$ and $\bar{\epsilon}$ small enough)
such that $(1 - \epsilon_1) H_{P_{XY} || P_X \times P_Y}(\calQ_X \times \calQ_Y) - \delta_2 > L$.
Thus the right side of \eqref{eq:nr_ld_ub} is less than $-L$ for any positive $L$.
\endproof

\myproof{Proof of Theorem  \ref{thm:nr_ld_lb}:}
The case that $D(P_{XY} || P_X \times P_Y) = \infty$ is trivial, thus we only consider
finite $D(P_{XY} || P_X \times P_Y)$.

We first show that for each $\epsilon > 0$, there is a finite partition $\calQ_X = \{A_i\}$ 
and $\calQ_Y = \{B_j\}$ of $E_X$ and $E_Y$, and $\lambda(\epsilon)$  such that

\begin{align}
 \left \{ P^{n}_{XY} \in B(P_{XY}, \epsilon) \right \} \supset
 \bigcap_{i,j} \left \{ |P^{n}_{XY}(A_i \times B_j) - P_{XY}(A_i \times B_j)| < \lambda(\epsilon) \right \}.
\end{align}
and $\lambda(\epsilon) \rightarrow 0$ as $\epsilon \rightarrow 0$.

To see this, let $\calF_X$ and $\calF_Y$ be fields as described in Corollary \ref{cor:field}.
For $k = 1, 2, \ldots$, let $\calQ_X^k \subset \calF_X$ be a sequence of successively finer finite partitions
of $E_X$ in the sense that if $A \in \calQ_X^k$ then $A$ is the union of  atoms of $\calQ_X^{k+1}$ 
and for each $A \in \calF_X$, $A$ is the union of atoms of $\calQ_X^k$ for some $k$. 
Likewise for $\calQ_Y^k \subset \calF_Y$.
We denote the atoms of $\calQ_X^k$ by $A_i^k$, $i = 1, \ldots, |\calQ_X^k| \eqdef N_X^k$, and 
likewise for the atoms $B_j^k$ of $\calQ_Y^k$.

Consider any sequence $M_{XY}^k$ of distributions that satisfy the sequence of events
\begin{align}
 \bigcap_{i,j} \left \{ |M^{k}_{XY}(A^k_i \times B^k_j) - P_{XY}(A^k_i \times B^k_j)| < \lambda_k \right \} \label{eq:forced}
\end{align}
where $\lambda_k \eqdef 1 / (k \times N^k_X \times N^k_Y)$. 
Then for any $A \in \calF_X$ and $B \in \calF_Y$, $M^k_{XY}(A \times B) \rightarrow P_{XY}(A \times B)$.
Thus by \cite[Chapter 1, Theorem 2.2]{Billingsley:1999}, the sequence of events in \eqref{eq:forced}
implies $P_{XY} = \wlim_k M_{XY}^{k}$. This implies
$\lim_k d(M_{XY}^{k},  P_{XY}) = 0$. Let $\epsilon_k =
\sup d(M_{XY}^{k},  P_{XY})$, where the supremum is over
$M_{XY}^{k}$ such that \eqref{eq:forced} holds at the $k$th step. 
We must have $\epsilon_k \rightarrow 0$ or there would be a choice
of $M_{XY}^{k}$ satisfying \eqref{eq:forced} such that 
$P_{XY} = \wlim_k M_{XY}^{k}$ does not hold.
Let $K$ be such that $\epsilon_K < \epsilon$. 

Hence, we can pick $\calQ_X = \calQ_X^K$, $\calQ_Y = \calQ_Y^K$ and 
any $\lambda(\epsilon) \leq \lambda_K$.
Therefore, with these choices,
\begin{align}
 &\liminf_{n \rightarrow \infty} \frac{1}{n} \log 
     \mu^{n}( (\bfX^{n}, \bfy^{n}) \in A^{n}_{\epsilon}(P_{XY}) ) \nonumber \\
 &\quad \geq
 \liminf_{n \rightarrow \infty} \frac{1}{n} \log 
     \mu^{n} \left(\bigcap_{i,j} \left \{ |P^{n}_{XY}(A_i \times B_j) - P_{XY}(A_i \times B_j)| < \lambda(\epsilon) \right \}.\right) \\
%
 &\quad \geq^{(a)}
 - D(P_{XY} || P_X \times P_Y) - \delta,
\end{align}
where inequality $(a)$ is justified below. 

To justify inequality $(a)$, first let $Q_{XY}$, $Q^n_{XY}$, etc denote the
appropriate induced distributions on the partitions of $\calQ_X$ and $\calQ_Y$
as in the proof of Theorem \ref{thm:nr_ld_ub}. Then, for any $\alpha > 0$,
\begin{align}
 \left| Q_{X|Y}^n(A|B) - Q_{X|Y}(A|B) \right| < \alpha \label{eq:epsilon}
\end{align}
implies
\begin{align}
 \left| Q_{XY}^n(A \times B) - Q_{XY}(A \times B) \right| < \alpha Q_Y(B) + \left|Q_Y^n(B) - Q_Y(B) \right|.
\end{align}
If $\bar{\epsilon}(\epsilon, \delta)$ is sufficiently small that 
$\bfy^n \in A_{\bar{\epsilon}(\epsilon, \delta)}^n(P_Y)$
implies $\left|Q_Y^n(B) - Q_Y(B) \right| < \alpha$ for all $B \in \calQ_Y$ with $\alpha < \lambda(\epsilon)/2$,
then the event
\begin{align}
\bigcap_{i,j} \{ \left| Q_{X|Y}^n(A_i|B_j) - Q_{X|Y}(A_i|B_j) \right| < \alpha \}
\end{align}
is a subset of
\begin{align}
\bigcap_{i,j} \left \{ |P^{n}_{XY}(A_i \times B_j) - P_{XY}(A_i \times B_j)| < \lambda(\epsilon) \right \}.
\end{align}
Thus
\begin{align}
 &\liminf_{n \rightarrow \infty} \frac{1}{n} \log 
     \mu^{n} \left(\bigcap_{i,j} \left \{ |P^{n}_{XY}(A_i \times B_j) - P_{XY}(A_i \times B_j)| < \lambda(\epsilon) \right \}.\right) \nonumber \\
 &\quad \geq
 \liminf_{n \rightarrow \infty} \frac{1}{n} \log 
     \mu^{n} \left( \bigcap_{i,j} \{ \left| Q_{X|Y}^n(A_i|B_j) - Q_{X|Y}(A_i|B_j) \right| < \alpha \}  \right) \\
 &\quad \geq -\sum_j (Q_Y(B_j) + \alpha) D(Q_{X|Y}(\cdot|B_j) || Q_X(\cdot) ) \\
 &\quad = -H_{P_{XY} || P_X \times P_Y}(\calQ_X \times \calQ_Y) - \alpha \sum_j D(Q_{X|Y}(\cdot|B_j) || Q_X(\cdot) ) \\
 &\quad \geq -D(P_{XY} || P_X \times P_Y) - \delta, \label{eq:nr_ld_lb}
\end{align}
where $\delta = \alpha \sum_j D(Q_{X|Y}(\cdot|B_j) || Q_X(\cdot) )$ is
finite\footnotemark as $D(P_{XY} || P_X \times P_Y) < \infty$. Furthermore,
$\delta$ can be made arbitrarily small by choosing $\alpha$ small enough,
which can be assured by choosing $\bar{\epsilon}(\epsilon, \delta)$ small enough.

\footnotetext{If $Q_Y(B_j) > 0$ then $D(Q_{X|Y}(\cdot|B_j) || Q_X(\cdot) )$ must be finite
as otherwise, $D(P_{XY} || P_X \times P_Y) \geq D(Q_{XY} || Q_X \times Q_Y) = \infty$. If $Q_Y(B_j) = 0$, then
$Q_{X|Y}(\cdot|B_j)$ is arbitrary, and selecting $Q_{X|Y}(\cdot|B_j) = P_X(\cdot)$
results in $D(Q_{X|Y}(\cdot|B_j) || Q_X(\cdot) ) = 0$.}
\endproof

\bibliographystyle{IEEEtranS}
\bibliography{mitran2}

\end{document}